\documentclass[journal]{IEEEtran}

\usepackage{amsmath}    
\usepackage{graphics,epsfig,subfigure}    
\usepackage{verbatim}   
\usepackage{color}      
\usepackage{subfigure}  
\usepackage{hyperref}   
\usepackage{amssymb}
\usepackage{epstopdf}
\usepackage{graphicx}
\usepackage{wasysym}
\usepackage{wrapfig}
\usepackage{sidecap}
\usepackage{float}
\usepackage{subfigure}
\usepackage{enumerate}
\usepackage{graphicx}
\usepackage{amsthm}
\graphicspath{{./figures/}}
\usepackage{algorithmic}
\usepackage{algorithm}
\usepackage{multirow}
\usepackage{mathrsfs}

\newtheorem{theorem}{Theorem}
\newtheorem{lemma}{Lemma}

\newtheorem{corollary}{Corollary}

\begin{document}
%
\title{Combinatorial Auction-Based Pricing for Multi-tenant Autonomous Vehicle Public Transportation System}

\author{Albert Y.S. Lam
\thanks{A.Y.S. Lam is with the Department of Electrical and Electronic Engineering, The University of Hong Kong, Pokfulam, Hong Kong (e-mail: ayslam@eee.hku.hk).}
}



\maketitle

\begin{abstract}
A smart city provides its people with high standard of living through advanced technologies and transport is one of the major foci. With the advent of autonomous vehicles (AVs), an AV-based public transportation system has been proposed recently, which is capable of providing new forms of transportation services with high efficiency, high flexibility, and low cost. For the benefit of passengers, multitenancy can increase market competition leading to lower service charge and higher quality of service. In this paper, we study the pricing issue of the multi-tenant AV public transportation system and three types of services are defined. The pricing process for each service type is modeled as a combinatorial auction, in which the service providers, as bidders, compete for offering transportation services. The winners of the auction are determined through an integer linear program. To prevent the bidders from raising their bids for
higher returns, we propose a strategy-proof Vickrey-Clarke-Groves-based charging mechanism, which can maximize the social welfare, to settle the final charges for the customers. We perform extensive simulations to verify the analytical results and evaluate the performance of the charging mechanism.
\end{abstract}
\IEEEpeerreviewmaketitle

\section{Introduction}

\IEEEPARstart{W}{ith} limited resources and high population density, people of many cities suffer from deteriorating living environments, health problems, traffic congestion, and pollution. To improve the standard of living, we may turn a city into a smart city through smarter utilization of resources with modern technologies \cite{smartcity}. Transport is one of key sectors in smart city research and development. The future transportation system should be able to accommodate the massive volume of passengers, support a broader range of services, and  lessen its impact on the environment.  To design an intelligent transportation system, we may contemplate the infrastructure, the vehicles, and the supportive management system. In most of the well-developed cities, there is not much room to have large modifications to the existing infrastructure and developing infrastructure usually involves a long time span. To facilitate widely applicable transportation in the near future, we may  focus on the vehicles and the management system. We may utilize various advanced vehicular technologies to enhance the performance of a transportation system. As being able to accommodate more passengers, we will focus on public transport in this paper.

Autonomous vehicles (AVs) refer to those vehicles capable of driving themselves without human intervention. They can adapt and respond to various situations happened on the roads due to their strong sensing and self-control abilities. They possess many advantages, including fewer traffic collisions, increased roadway capacity, alleviation of parking scarcity, and reduction in car theft \cite{AVWiki}. Many companies have already invested in the AV industry. Google launched the Self-Driving Car Project in 2011 \cite{google}. Automotive manufacturers, like BMW \cite{BMW} and Mercedes-Benz \cite{benz}, are inventing their AVs for mass production. Moreover, related law has been passed in Nevada, Florida, California, and Michigan to allow AVs driving on public roads \cite{AVlegal}. All these show that the AV is a promising technology and it will become one of the major elements in the future transportation system.

Recently a new AV-based public transportation system was proposed in \cite{confversion}, where AVs are utilized as the conveyances to carry passengers in a city. The system can be made automatic and adaptive to transportation requests, with the consideration of traffic conditions. There is a control center (e.g., a central computing facility) for managing and scheduling the AVs, and making other related decisions for the system. Through adaptive scheduling, the optimal routes with the minimum operational cost can be determined for the AVs. With advanced vehicular communications technologies, like vehicular ad-hoc networks \cite{vanet_survey},  the AVs become connected able to exchange information with the control center.
Due to the unmanned nature, the AVs can coordinate with each other, and loyally and accurately follow the instructions from the control center. The system allows ride sharing; passengers may share their rides with other people, as buses. Moreover, the system supports point-to-point services; a passenger can specify the pick-up and drop-off locations with various time requirements, as taxis. Hence the AV public transportation system provides new forms of transportation services with high efficiency, high flexibility, and low cost.

Increasing market competition can lead to better prices, quality of services, and information for consumers with more choices.
For the benefit of its residents, a smart city may allow its AV public transportation system to be operated by multiple service providers, that results in a multi-tenant system. In such a system, there is still one control center for information gathering and central decision making, but the AVs are divided into groups, each of which is governed by an independent service provider (or called an operator). While the technical scheduling and admission control mechanisms have been fully addressed in \cite{confversion} and \cite{admissioncontrol}, there are no clear rules about how to set the service charges for the system, especially the multi-tenant system. In this paper, we aim to study the pricing issue for the multi-tenant system and propose a combinatorial auction-based pricing scheme, where the service charges are determined through the Vickrey-Clarke-Groves (VCG) mechanism \cite{Vickrey,Clarke,Groves}. Our scheme can maximize the social welfare and it is strategy-proof such that all bidders (i.e., the operators) have no incentive to lie about their private information in the auction.

Auctions have been widely used to settle prices for many engineering systems. In \cite{Auction_Elec}, the uniform pricing and pay-as-bid pricing were settled with a multi-unit auction in a short-term electricity market. \cite{doublelayer} constructed a double auction to determine the price with quantity for energy trading in a vehicle-to-grid system. In \cite{Auction_internet}, a multibid auction scheme was designed to allow users to compete for bandwidth in telecommunication networks. \cite{Auction_cognitive} designed a multi-unit sequential sealed-bid first-price auction for spectrum trading in cognitive radio networks.
Due to its desirable properties, VCG is a powerful mechanism in mechanism design. Samadi \textit{et al.} proposed a VCG-based mechanism to maximize the aggregate utility of all users and minimize the power generation cost for demand side management in the smart grid \cite{VCG_smartgrid}. \cite{VCG_Kelly} proposed a VCG-based resource allocation mechanism with the Kelly mechanism so that efficiency is attained at Nash equilibrium points. In \cite{VCG_cloud}, the VCG mechanism was utilized to procure resources and select cloud vendors in cloud computing.
We can see that, with many successful applications, a VCG-auction-based mechanism can be utilized to address the pricing issue of the multi-tenant AV public transportation system. 


The rest of this paper is organized as follows. We present the system model in Section \ref{sec:model}. We construct the combinatorial auction for the system  in Section \ref{sec:auction} and Section \ref{sec:charging} discussed the VCG-based charging mechanism. We verify the analytical results and evaluate \textcolor{black}{the}  system performance in Section \ref{sec:simulation}. Finally we conclude this paper in Section \ref{sec:conclusion}.



\section{System Model}\label{sec:model}

In this section, we first introduce the model of the AV public transportation system and then pinpoint its pricing process.

\subsection{Autonomous Vehicle Public Transportation System}\label{subsec:model}

The AV public transportation system was firstly proposed in \cite{confversion}, in which the control center coordinates a fleet of AVs to provide transportation services. Customers submit transportation requests to the control center with necessary information, including pickup and dropoff locations, service times, number of passengers, etc. After collecting a number of transportation requests, the control center assigns appropriate AVs to serve the requests. When serving the requests, the AVs may be carrying the passengers from other requests and this realizes ride-sharing. Due to the unmanned nature of the vehicles, we need to determine the schedules and routes for the designated AVs in order to admit the transportation requests. This can be accomplished by solving the scheduling problem, which was fully addressed in \cite{confversion,admissioncontrol}. Moreover, not every request can be served because the requirements stated in the request may not simultaneously be satisfied by any vehicle available. We need to perform admission control to screen out those infeasible requests for effective scheduling. The admission control problem is given as a bi-level optimization, in which the scheduling problem is considered as a constraint. The details can be found in \cite{admissioncontrol}.

\begin{figure}[!t]
\centering
\hspace{-0.3cm}
\includegraphics[width=3.3in]{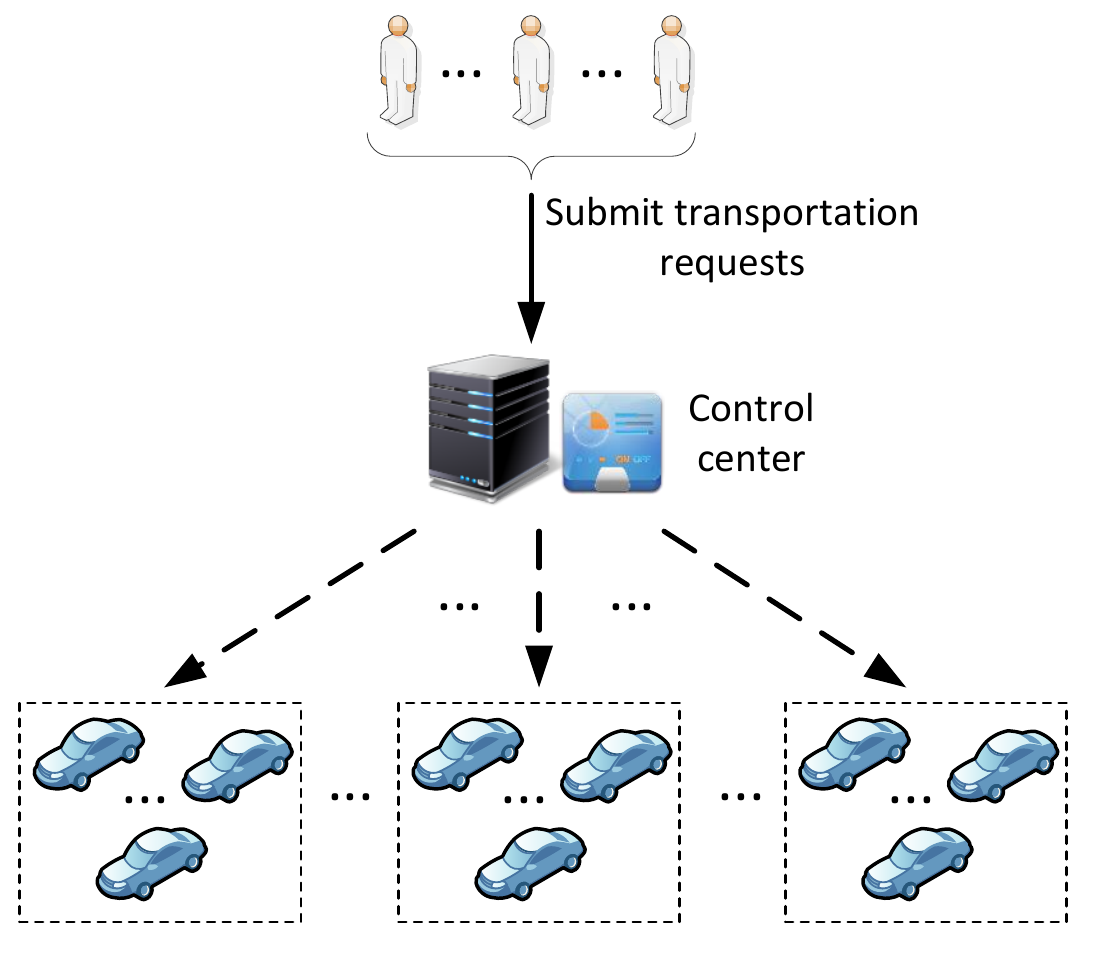}
\caption{Multi-tenant AV public transportation system.} 
\label{fig:system}
\end{figure}

\cite{confversion} and \cite{admissioncontrol} focused on the technical aspects of the system, in which \textcolor{black}{a} monopoly is simply assumed to operate the system. In this paper, we generalize the operation model by supporting multi-tenancy. In a multi-tenant AV public transportation system, multiple fleets of AVs are operated and managed by independent operators.\footnote{For implementation, we can consider that the control center is operated in a cloud \cite{cloud} and the AV operators are also engaged in the same or different clouds.} The vehicles under the same tenant are cooperative but those belonging to different tenants are not. For the transparency of information and benefits of customers, we design a central ``market'' to convene all the received transportation requests. This market is run by a broker who aims to match the service ``buyers'' and ``sellers''. In other words, upon a service request $r$ has been submitted by a customer, the broker tries to see if there is any operator able to provide the service. The operators will assess $r$ by checking its admissibility and the corresponding operational cost\footnote{The operational cost for serving a request is determined through admission control and scheduling and the methodologies discussed \cite{admissioncontrol} can still be carried through in the multi-tenant system.} and this will result in one of the following cases: (1) no operator can admit $r$; (2) only one operator decides to admit $r$; and (3) more than one operators intend to ``sell'' their service for $r$. Cases 1 and 2 are simple. For Case 1, we may simply ignore $r$ or keep it on hold for later consideration. For Case 2, the winning operator will schedule an AV to serve $r$ based on the result of scheduling. So the complication mainly comes from Case 3 where the broker should decide one appropriate operator to complete the deal. Hence we need a mechanism to select one of the competing operators based on some criteria and designing such an mechanism will be the focus of this paper.

As discussed in \cite{confversion,admissioncontrol}, in a monopolistic system, there is a control center responsible for information gathering and decision making. In fact, in a multi-tenant system, the control center can act as the broker to arbitrate admission competition, i.e., Case 3 above. 
The system is depicted in Fig. \ref{fig:system}. Customers submit transportation requests to the control center through any appropriate means (e.g., mobile apps and phone calls) and each submitted request contains all the necessary information about the service, e.g., service starting point and destination, service start time and end time, the number of seats occupied, etc. 
When receiving a transportation request from a customer, the control center disseminates the information to the operators. Based on the methods discussed in \cite{confversion,admissioncontrol}, each operator can independently assess the request according to the conditions of its governed vehicles,  estimate the operational cost for providing the service, and decide which vehicle in its fleet is carrying out the service.
Based on the estimated operational cost, the operator will list its proposed service charge at the control center. 

\subsection{Pricing Process} \label{subsec:pricingprocess}

\begin{figure}[!t]
\centering
\hspace{-0.3cm}
\includegraphics[width=3.5in]{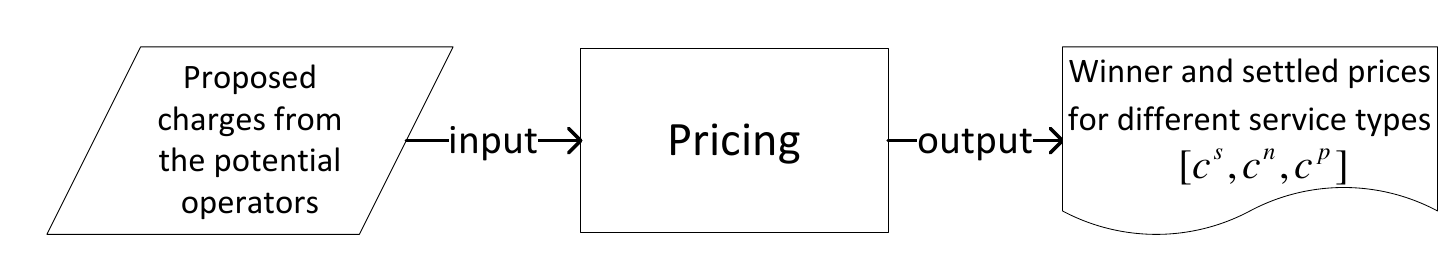} 
\caption{The pricing process.} 
\label{fig:pricing}
\end{figure}

The pricing process of the multi-tenant system is manipulated by the control center  and its schematic can be found in Fig. \ref{fig:pricing}. 
\textcolor{black}{Based on \cite{admissioncontrol}, the system operates in an interval basis, where the pricing process takes place in the duty assignment sub-interval, following the schuduling and admission control processes.}
After collecting the proposed charges from all potential service providing operators, the control center performs the pricing process and determines the winner(s) of the competition with a set of settled prices for the service.\footnote{The reason why returning a set of settled prices instead of just one price will be explained next.}
Suppose that all vehicles are homogeneous in terms of capacity and each has $Q$ seats. Consider that the number of seats required for request $r$ is $q_r$.\footnote{We call a person who submits the request a customer, and who takes the ride a passenger. For request $r$, there is one customer and $q_r$ passengers. A customer is not necessarily one of the passengers.} Assuming $1\leq q_r\leq Q$, 
normally the customer is interested in three types of services:
\begin{itemize}
	\item Splittable service: The service is allowed to be supported by more than one vehicle. In this way, the passengers of the request will be split into multiple groups and travel on separate vehicles.  Yet we do not exclude the situation that all passengers may stay along in the same vehicle in some cases. We denote the charge of the splittable service by $c^s$.
	\item Non-splittable service: The service is accomplished by one single vehicle only. This accounts for the situation that the passengers prefer to stay along together during the journey. The charge of the non-splittable service is denoted by $c^n$.
	\item Private service: The service is solely supported by one vehicle and no other passengers of other requests are allowed to stay in the vehicle during the time serving $r$ even though there are not enough passengers to occupy the whole vehicle. In other words, the passengers desire to hire a private vehicle for the travel. We denote the charge of the private service by $c^p$.
\end{itemize}
With the values of $c^s$, $c^n$, and $c^p$, the customer can then make the final decision of adopting which service type based on its own preference. 
 For example, if $c^s$ is much lower than $c^n$ and $c^p$, the customer may make a cost-effective decision and go to the splittable service. If $c^p$ has a similar value as of $c^s$ and $c^n$, the customer may choose the private service for privacy.
Note that the customer does not need to take care of which operator(s) or vehicle(s) will eventually provide the service. The pricing process will consolidate all the proposed charges from the operators and determine the best service prices for the customer,  
who will only see the settled prices of the three service types. 
Based on the customer's decision, the control center will match the operator(s) for the customer. The involved operator(s) will then carry out the plan for providing the service according to the schedules resulted from the prior scheduling process.

\textcolor{black}{Note that each pricing process is supposed to determine the service charges for one service request. Since the system operates in an interval basis, there may be multiple service requests being handled at the same time in each interval. In fact, after scheduling and admission control \cite{admissioncontrol}, the service requests which are being processed are no longer dependent and thus we can simply apply the pricing process to each of these service requests and compute the results in parallel simultaneously.}

\section{Combinatorial Auction} \label{sec:auction}

We model the pricing process with combinatorial auctions, one for each service type.\footnote{Note that the design of the combinatorial auction in this section and the charging mechanism in the next section do not rely on the the technical details of  \cite{confversion} and \cite{admissioncontrol}.} In a combinatorial auction \cite{CombAuctionSurvey}, there are a number of discrete items for sale or asking for sellers. Bidders can place bids on combinations of the items instead of individual items. An auctioneer facilitates the auction and decides the winners based on the submitted bids. If the bidders are buyers, for the benefit of the owner of the items, the bidders with the highest bids will win. On the other hand, if the bidders are sellers, for the welfare of the item requester, the bidders with the lowest bids will be the winners. 
For the pricing process, seat occupancies are the items to be traded. For request $r$, we need to assign $q_r$ seats from various AVs to the passengers. We can interpret the customer of request $r$ as the ``buyer'' of the transportation service and it needs $q_r$ seats to be `sold''. Those operators who are interested in ``selling'' their seats are the bidders and they place bids on combinations of seats. The control center acts as the auctioneer to facilitate the auction.

Consider request $r$ with $q_r$ seats required and there are $K$ operators interested in serving $r$.  Let $\mathcal{K}$ be the set of bidders and hence $|\mathcal{K}|=K$. Assume that both the buyer and sellers are rational; for the same service, the customer aims for the lowest price while the operators try to win other bidders for offering the service. 
All vehicles governed by the operators are homogeneous such that their capacities are all identical.\footnote{Since AVs are unmanned, the vehicular specifications should follow some standards which be strictly regulated by the government. Also vehicles of similar size have advantages of simpler control and interaction between vehicles. Moreover, it is likely that the operators adopt similar models of AVs for business. Thus, for simplicity, we assume that all vehicles are homogeneous.} Let $Q$ be the capacity of each vehicle. 
Without loss of generality, we assume 
\begin{align} \label{assumption}
1\leq q_r\leq Q.
\end{align}
As discussed in \cite{confversion}, a request with the number of seats required larger than $Q$ can be considered as multiple requests such that they all require the seat quantities smaller than $Q$.
%
We also assume that all seats are homogeneous; in a vehicle, the passengers only care about the number of seats available, instead of the positions and conditions of the seats. 
We define a set $\mathcal{Q} = \{\{1\},\{1,2\},\ldots,\{1,2,\ldots,m\},\ldots,\{1,2,\ldots,Q\}\}$ and it is the collection of seat combinations for lease, where $\{1,2,\ldots,m\}$ represents a set of $m$ seats. 
%
For each combination $\mathcal{S}\in \mathcal{Q}$, bidder $k\in\mathcal{K}$ has a valuation $v_k(\mathcal{S})$ for letting $\mathcal{S}$. This valuation is the base price for the lease. If the customer pays less than $v_k(\mathcal{S})$, operator $k$ will induce a negative utility for letting $\mathcal{S}$. Recall that bidder $k$ can determine the cost of offering $\mathcal{S}$ from the scheduling and admission control processes. It may set $v_k(\mathcal{S})$ by marking up the cost with a certain guaranteed return.
In the auction, bidder $k$ places a bid $b_k(\mathcal{S})$ for each $\mathcal{S}\in \mathcal{Q}$. Occupying a seat always induces a charge.
This can be realized by making $b_k(\cdot)$ an increasing function with the size of seat combination; for $\mathcal{S},\mathcal{T}\in \mathcal{Q}$, if $|\mathcal{S}|<|\mathcal{T}|$, then $b_k(\mathcal{S})<b_k(\mathcal{T})$.
Since the vehicles may be serving other passengers at the time of auctioning off request $r$, the numbers of seats available in various vehicles can be different. Let $Q_k$ be the number of seats offerable by bidder $k\in \mathcal{K}$. If $\mathcal{S}$ corresponds to the seats with quantity larger than that available, i.e., $|\mathcal{S}|>Q_k$, we can simply set its valuation $v_k(\mathcal{S})$ and bid $b_k(\mathcal{S})$ with sufficiently large numbers such that $k$ will not win $\mathcal{S}$ in the auction. This is equivalent to the fact that $\mathcal{S}$ will never be allocated to $k$.
Bidder $k$ may make more profit by raising $b_k(\mathcal{S})$ higher than $v_k(\mathcal{S})$. In general, we can assume $b_k(\mathcal{S})\geq v_k(\mathcal{S})\geq 0$, for all $\mathcal{S}\in \mathcal{Q}$. 
Moreover,  each bidder only concern about what it will receive in the auction but not the results of other bidders. 

We construct one auction for each service type. Next we investigate how the auctioneer determines the winners of the auctions and we model these decision making processes with integer linear programs. We define the binary variable $x_k(\mathcal{S})$ to indicate if $\mathcal{S}$ is allocated to $k$, i.e., 
\begin{align*}
x_k(\mathcal{S})=\left\{
	\begin{array}{ll}
		1 & \text{if $\mathcal{S}$ is allocated to $k$,}\\
		0 & \text{otherwise.}
	\end{array}
\right.
\end{align*}
For the benefit of the customer, the service should be provided by those operators charging the least. The objective is to minimize the total bids charged by the operators and the corresponding objective function is
\begin{align} \label{objective}
p=\sum_{k\in\mathcal{K}}\sum_{\mathcal{S}\in \mathcal{Q}}b_k(\mathcal{S})x_k(\mathcal{S}).
\end{align}
$p$ is the optimal payment that the customer should pay for the service. For some reasons (explained in Section \ref{sec:charging}), the customer will be charged with a higher amount. We will discuss how to determine the final charge from $p$ in Section \ref{sec:charging}.
\begin{lemma} \label{onebid}
With an increasing concave bidding function, each bidder can at most win with one bid in the combinatorial auction for each service type.
\end{lemma}
\begin{proof}
Suppose bidder $k$ wins two bids, e.g., for $\mathcal{S},\mathcal{T}\in\mathcal{Q}$, in the auction. It is impossible to have $|\mathcal{S}|+|\mathcal{T}|>Q$. By \eqref{assumption}, the customer will never request a  service which provides unnecessary seats with avoidable extra cost. Now we focus on $\mathcal{S}$ and $\mathcal{T}$ with $|\mathcal{S}|+|\mathcal{T}|\leq Q$. Recall that seats are unidentifiable while the elements in $\mathcal{Q}$ are identified with the seat quantity only. So we will always find $\mathcal{U}$ in $\mathcal{Q}$ such that $|\mathcal{U}|=|\mathcal{S}|+|\mathcal{T}|$. Since $b_k(\cdot)$ is an increasing concave function, $b_k(\mathcal{U}) < b_k(\mathcal{S})+b_k(\mathcal{T})$. As we are minimizing \eqref{objective}, the customer will adopt $\mathcal{U}$ instead of the combination of $\mathcal{S}$ and $\mathcal{T}$. This induces a contradiction. Hence, each bidder can at most win with one bid.
\end{proof}
In general, the marginal cost of seat occupancy is diminishing. When a vehicle has been assigned to traverse from one place to another, the extra cost of holding one more passenger is rather limited if the route does not need to be altered. In this case, $b_k(\cdot)$ is an increasing concave function, i.e., for $m\leq Q-2$, $\mathcal{S}=\{1,\ldots,m\},\mathcal{S}'=\{1,\ldots,m,m+1\}$, and $\mathcal{S}''=\{1,\ldots,m,m+1,m+2\}$, we have $b_k(\mathcal{S}')-b_k(\mathcal{S})\geq b_k(\mathcal{S}'')-b_k(\mathcal{S}')$.

However, if the bidding function is not concave, Lemma \ref{onebid} will not not hold. To ensure a seat not be let more than once at the same time, we explicitly impose the following constraint:
\begin{align} \label{lemmacons}
	\sum_{\mathcal{S}\in\mathcal{Q}} x_k(\mathcal{S}) \leq 1, \forall k\in \mathcal{K}.
\end{align}
We will utilize \eqref{lemmacons} to formulate the optimization problems for the splittable, non-splittable, and private services.

\subsection{Splittable Service}
For the splittable service, once all bids have been submitted, the auctioneer determines the winner(s) by solving the following optimization problem:
\begin{subequations} \label{WDP_ps}
\begin{align}
\text{minimize}\quad 	& \sum_{k\in\mathcal{K}}\sum_{\mathcal{S}\in \mathcal{Q}}b_k(\mathcal{S})x_k(\mathcal{S}) \label{ps_obj}\\
\text{subject to}\quad 
& \sum_{\mathcal{S}\in\mathcal{Q}} x_k(\mathcal{S}) \leq 1, \forall k\in \mathcal{K}\label{ps_con1}\\
& \sum_{k\in\mathcal{K}}\sum_{\mathcal{S}\in \mathcal{Q}}|\mathcal{S}|x_k(\mathcal{S})\geq q_r, \label{ps_con2}\\
& x_k(\mathcal{S}) =\{0,1\},\forall \mathcal{S}\in\mathcal{Q}, k\in \mathcal{K} \label{ps_con3}.
\end{align}
\end{subequations}
Constraint \eqref{ps_con2} guarantees that the offers have enough seats as stated in the transportation request while Constraint \eqref{ps_con3} confines that $x_k(\mathcal{S})$ is a binary variable. Note that \eqref{WDP_ps} may result in an offer that all seats come from the same vehicle as long as the offer incurs the lowest bid in total.
\begin{lemma}
For the splittable service, the number of seats allocated in the auction must be equal to that stated in the transportation request.
\end{lemma}
\begin{proof}
For request $r$, the number of seats required is $q_r$. Suppose that $\tilde{x}=[\tilde{x}_k(\mathcal{S})]_{k\in\mathcal{K},\mathcal{S}\in\mathcal{Q}}$ is a global minimum  of Problem \eqref{WDP_ps} such that 
\begin{align*}
\sum_{k\in\mathcal{K}}\sum_{\mathcal{S}\in \mathcal{Q}}|\mathcal{S}|\tilde{x}_k(\mathcal{S})> q_r.
\end{align*}
We construct another $x^*$ from  $\tilde{x}$ as follows: 
\begin{itemize}
\item For each $k\in\mathcal{K}$, if $\sum_{\mathcal{S}}\tilde{x}_k(\mathcal{S})=0$, then $\sum_{\mathcal{S}}x^*_k(\mathcal{S})=0$;
\item Among those $k$ and $\mathcal{S}$ with $\tilde{x}_k(\mathcal{S})=1$, for some of the $k$, we set $x^*_k(\mathcal{T})=1$  on $\mathcal{T}$ instead of $\mathcal{S}$, where $\mathcal{T}$ is a smaller set with $|\mathcal{T}|<|\mathcal{S}|$, while for the other $k$, we set $x^*_k(\mathcal{S})=1$. This assignment of $x^*$ needs to satisfy
\begin{align} \label{equalcon3}
\sum_{k\in\mathcal{K}}\sum_{\mathcal{S}\in \mathcal{Q}}|\mathcal{S}|x^*_k(\mathcal{S})= q_r.
\end{align}
\end{itemize}
The above construction of $x^*$ ensures that \eqref{ps_con1} and \eqref{ps_con3} are satisfied. Hence $x^*$ is also a feasible solution of \eqref{WDP_ps}. Since $b_k(\cdot)$ is an increasing function and $x^*$ is set on some smaller $\mathcal{T}$, $\sum_{k}\sum_{\mathcal{S}}b_k(\mathcal{S})x^*_k(\mathcal{S})$ is smaller than $\sum_{k}\sum_{\mathcal{S}}b_k(\mathcal{S})\tilde{x}_k(\mathcal{S})$. Hence $\tilde{x}$ cannot be a global minimum and this induces a contradiction. Therefore. a global optimum $x^*$ of \eqref{WDP_ps} should satisfy \eqref{equalcon3}, which is equivalent to the statement of the lemma. 
\end{proof}
From this lemma, we can make \eqref{ps_con2} an equality so that the feasible region of \eqref{WDP_ps} becomes smaller. However, in general, most existing integer programming solvers tackle problems with inequality more effectively than with equality. That is the main reason that we relax  \eqref{ps_con2} as an inequality.

\subsection{Non-splittable Service}
Similar to the splittable service, the auctioneer determines the winner of the auction for the non-splittable service by addressing the following problem:
\begin{subequations} \label{WDP_pn}
\begin{align}
\text{minimize}\quad 	& \sum_{k\in\mathcal{K}}\sum_{\mathcal{S}\in \mathcal{Q}}b_k(\mathcal{S})x_k(\mathcal{S}) \\
\text{subject to}\quad & \sum_{\mathcal{S}\in \mathcal{Q}}\sum_{k\in\mathcal{K}}x_k(\mathcal{S}) \leq 1, \label{pn_con1}\\
& \sum_{k\in\mathcal{K}}\sum_{\mathcal{S}\in \mathcal{Q}}|\mathcal{S}|x_k(\mathcal{S})\geq q_r,\label{pn_con2}\\
& x_k(\mathcal{S}) =\{0,1\},\forall \mathcal{S}\in\mathcal{Q}, k\in \mathcal{K}.
\end{align}
\end{subequations}
\eqref{WDP_pn} is similar to \eqref{WDP_ps} except \eqref{pn_con1}, which ensures that the offer will be accomplished with at most one subset of seats in $\mathcal{Q}$ from all bidders. In fact, \eqref{pn_con1} implies \eqref{ps_con1}. In this way, the passengers will be not split into multiple vehicles for the travel.

\subsection{Private Service}
For the private service, the offer should be originated from one single vehicle, in which no other customers have been pre-assigned. The setting is similar to that for a non-splittable ride, but the number of seats required is equal to the capacity of a vehicle, i.e., $Q$. The auctioneer can determine the winner with the following problem:
\begin{subequations} \label{WDP_pw}
\begin{align}
\text{minimize}\quad 	& \sum_{k\in\mathcal{K}}\sum_{\mathcal{S}\in \mathcal{Q}}b_k(\mathcal{S})x_k(\mathcal{S})\\
\text{subject to}\quad & \sum_{\mathcal{S}\in \mathcal{Q}}\sum_{k\in\mathcal{K}}x_k(\mathcal{S}) \leq 1,\label{pw_con1} \\
& \sum_{k\in\mathcal{K}}\sum_{\mathcal{S}\in \mathcal{Q}}|\mathcal{S}|x_k(\mathcal{S})\geq Q,\label{pw_con2}\\
& x_k(\mathcal{S}) =\{0,1\},\forall \mathcal{S}\in\mathcal{Q}, k\in \mathcal{K}.
\end{align}
\end{subequations}
\eqref{WDP_pw} is similar to to \eqref{WDP_pn} by replacing $q_r$ with $Q$ in \eqref{pw_con2}. It ensures the offer will be made to $\mathcal{S}$ from an empty vehicle. 


\subsection{Results of the Auctions}

In combinatorial auctions, Problems \eqref{WDP_ps}, \eqref{WDP_pn}, and \eqref{WDP_pw} belong to the class of Winner Determination Problem \cite{CombAuctionSurvey}, but the formulations given above are dedicated to the multi-tenant AV system. Let $p^s$, $p^n$, and $p^p$ be the optimal objective function values of Problems \eqref{WDP_ps}, \eqref{WDP_pn}, and \eqref{WDP_pw}, respectively.

\begin{theorem} \label{lemma:pricerank}
The total bid required to accomplish the splittable service is the lowest while that for the private service is the highest, i.e., $p^s\leq p^n\leq p^p$.
\end{theorem}
\begin{proof}
$p^s$, $p^n$, and $p^p$ are the objective function values of the global optimums of Problems \eqref{WDP_ps}, \eqref{WDP_pn}, and \eqref{WDP_pw}, respectively. These minimization problems share an identical objective function. The result can be proved if we can show that the feasible region of \eqref{WDP_ps} embeds that of \eqref{WDP_pn}, which in turn embeds the feasible region of \eqref{WDP_pw}.

For \eqref{WDP_ps} and \eqref{WDP_pn}, their constraints are the same except \eqref{ps_con1} and \eqref{pn_con1}. It is not difficult to see that \eqref{pn_con1} is a special case of \eqref{ps_con1}. Hence the feasible region of \eqref{WDP_ps} embeds that of \eqref{WDP_pn}.

Similarly, for \eqref{WDP_pn} and \eqref{WDP_pw}, their constraints are the same except \eqref{pn_con2} and \eqref{pw_con2}. It is trivial to see that \eqref{pw_con2} is a special case of \eqref{pn_con2}. Thus the feasible region of \eqref{WDP_pn} embeds that of \eqref{WDP_pw}.
\end{proof}

Theorem \ref{lemma:pricerank} matches our intuition in the quality of service. If we desire to hire a private vehicle, the charge should be higher than that for a shared ride. If we do not mind to have the passengers being allocated on different vehicles, the charge should be lower than that of having all passengers reside on the same vehicle. It is because a splittable ride provides higher flexibility for the operators to schedule their vehicles with lower operational cost.

Note that in \eqref{WDP_ps}, \eqref{WDP_pn}, and \eqref{WDP_pw}, the winners of the auction are determined based on the submitted bids $b_k(\mathcal{S})$. Mathematically, they are constant and they are  solely assigned by the respective bidders. Thus the values of $p^s$, $p^n$, and $p^p$ depend on $b_k(\mathcal{S})$. We cannot directly set $p^s$, $p^n$, and $p^p$ to be the final charges because we have no way to prevent the bidders from manipulating their bids for possible higher profit.  If $b_k(\mathcal{S})$ is replaced by the true valuation $v_k(\mathcal{S})$ in \eqref{WDP_ps}, \eqref{WDP_pn}, and \eqref{WDP_pw}, the corresponding $p^s$, $p^n$, and $p^p$ will provide lower bounds on the payments that the customer needs to pay, as no bidder $k$  bids below $v_k(\mathcal{S})$. Apparently, it seems that the auctioneer cannot strike the payments toward these lower bounds for the three types of services. In the next section, we will design a charging mechanism to address this issue.

\section{Charging Mechanism} \label{sec:charging}

For \eqref{WDP_ps}, \eqref{WDP_pn}, and \eqref{WDP_pw}, the solutions are said to be economically efficient if the submitted bids are equal to the corresponding true valuations, i.e., $b_k(\mathcal{S})=v_k(\mathcal{S})$ for all $\mathcal{S}$. Economic efficiency provides the best allocation of seats for the benefit of the customer and an economically efficient allocation ensures that no Pareto improvement can be further achieved.
However, $b_k(\mathcal{S})$ is private information with respect to each bidder and the auctioneer cannot assert that the bidders have submitted the true valuations as the bids. A bidder may set $b_k(\mathcal{S})$ higher than $v_k(\mathcal{S})$ for better return but this will also lower its chance of winning the auction as another bidder may have set its bid lower than $b_k(\mathcal{S})$. In order to let the bidders behave truthfully, we adopt the VCG mechanism to determine the amount of payment that the customer should be charged for accepting the offer made from \eqref{WDP_ps}, \eqref{WDP_pn}, or \eqref{WDP_pw}.

\subsection{\textcolor{black}{VCG Mechanism}}

We demonstrate the VCG-based charging mechanism for the splittable service based on Problem \eqref{WDP_ps}, as an example, in the sequel. For the non-splittable and private services, the implementations are similar. 
Consider that bidder $k$ may not bid with the true valuation for a particular $\mathcal{S}$ and announce $b_k(\mathcal{S})\geq v_k(\mathcal{S})$. 
Let $\mathcal{X}$ be the feasible region of \eqref{WDP_ps}, $x^*$ an optimal solution of \eqref{WDP_ps}, and $p^*$ be its objective function value, i.e., 
\begin{align*}
	x^* &= \arg\min_x\{\sum_{k\in\mathcal{K}}\sum_{\mathcal{S}\in \mathcal{Q}}b_k(\mathcal{S})x_k(\mathcal{S})|x\in\mathcal{X}\},\\
	p^*&=\inf\{\sum_{k\in\mathcal{K}}\sum_{\mathcal{S}\in \mathcal{Q}}b_k(\mathcal{S})x_k(\mathcal{S})|x\in\mathcal{X}\}.
\end{align*}
Define 
\begin{align} \label{pminusk}
	p_{-k}^*\triangleq\inf\{&\sum_{l\in\mathcal{K}\setminus k}\sum_{\mathcal{S}\in \mathcal{Q}}b_l(\mathcal{S})x_l(\mathcal{S})|\sum_{\mathcal{S}\in \mathcal{Q}}x_l(\mathcal{S}) \leq 1, \forall l\in\mathcal{K}\setminus k,\nonumber \\
	&\sum_{l\in\mathcal{K}\setminus k}\sum_{\mathcal{S}\in \mathcal{Q}}|\mathcal{S}|x_l(\mathcal{S})\geq q_r, \nonumber\\
	&x_l(\mathcal{S}) =\{0,1\},\forall \mathcal{S}\in\mathcal{Q}, l\in \mathcal{K}\setminus k\}.
\end{align}
Note that the problem in \eqref{pminusk} is different from \eqref{WDP_ps} as $k$ is excluded from $\mathcal{K}$.
$p_{-k}^*$ is the ``welfare'' of the other bidders in the absence of bidder $k$. In other words, it is the result of the auction when bidder $k$ is absent.
Let $\bold{b}_k=[b_k(\mathcal{S})]_{\mathcal{S}\in\mathcal{G}}$ and $\bold{b}_{-k}=[b_l(\mathcal{S})]_{\mathcal{S}\in\mathcal{G},l\in\mathcal{K}\setminus k}$ be the set of bids submitted by $k$ and those submitted by all other bidders other than $k$, respectively, for all the seat combinations.
The VCG mechanism imposes that each bidder should charge the amount of ``damage'' it introduces to the community, i.e., the set of all operators.
We denote the total charge received by $k$ of providing the service based on the submitted bids by $c_k(\bold{b}_{-k})$. For each $k$, we set $c_k(\bold{b}_{-k})$ as the following:
\begin{align}
c_k(\bold{b}_{-k}) \triangleq p_{-k}^* - \sum_{l\in\mathcal{K}\setminus k}\sum_{\mathcal{S}\in \mathcal{Q}}b_l(\mathcal{S})x_l^*(\mathcal{S}). \label{totalcharge}
\end{align}
$c_k(\bold{b}_{-k})$ is interpreted as the ``social charge'' of $k$, which is equivalent to the difference of total welfare of other bidders without $k$'s involvement and that with $k$'s involvement.

Let $\bold{b}=(\bold{b}_k,\bold{b}_{-k})$ and consider the utility $u_k(\bold{b})$ of bidder $k$ as $u_k(\bold{b})=c_k(\text{b})-v_k(x(\text{b}))$, where $c_k(\text{b})$ is the charge received by $k$ and $v_k(x(\text{b}))$ is the sum of valuation of the service provided by $k$ based on $x\in\mathcal{X}$. So the utility is defined as the difference between the ``revenue'' and ``cost''. Each bidder $k$ tries to maximize its utility by manipulating its declared bids, as
\begin{align}
	\max_{\bold{b}_k} \left(c_k(\bold{b})-v_k(x(\bold{b}))\right).\label{utility}
\end{align}
Social welfare is the sum of the utilities of all bidders expressed as $\sum_{k}u_k(\bold{b})=\sum_{k}c_k(\text{b})-v_k(x(\text{b}))$.
\begin{theorem} \label{theorem:social}
Minimizing the total bid received by the operators, i.e., solving \eqref{WDP_ps}, \eqref{WDP_pn}, or \eqref{WDP_pw}, can maximize the social welfare if and only if all bidders submit their bids as their true valuations.
\end{theorem}
\begin{proof}
Let $b_k(x)$ be the sum of bids received by $k$ based on $x$.
By \eqref{totalcharge}, \eqref{utility} becomes
\begin{align*}
	\max_{\bold{b}_k}  \left(p_{-k} - \sum_{l\in\mathcal{K}\setminus k}b_l(x(\bold{b}))  -v_k(x(\bold{b}))\right).
\end{align*}
Since $p_{-k}$ is independent of $\bold{b}_k$, bidder $k$ will declare $\bold{b}_k$ based on
\begin{align}
	\max_{\bold{b}_k} \left(- \sum_{l\in\mathcal{K}\setminus k}b_l(x(\bold{b}))  -v_k(x(\bold{b}))\right). \label{utility3}
\end{align}
The only way that bidder $k$ influences the result of \eqref{utility3} is through the choice of $x$. In other words, it will declare $\bold{b}_k$ that leads the auctioneer to $x\in \mathcal{X}$ that solves
\begin{align}
	\min_{x\in \mathcal{X}} \left( \sum_{l\in\mathcal{K}\setminus k}b_l(x)  +v_k(x)\right). \label{utility4}
\end{align}
If $b_k(x)=v_k(x)$, for the splittable service, \eqref{utility4} is equivalent to \eqref{WDP_ps}. 
On the other hand, if \eqref{utility4} is equivalent to \eqref{WDP_ps}, $b_k(x)$ will be equal to $v_k(x)$ for all $k$. This completes the proof for the splittable service.
For the non-splittable (private) service, the proof is similar by comparing \eqref{utility4} with \eqref{WDP_pn} (\eqref{WDP_pw}).
\end{proof}
\begin{corollary}
Telling the truth, i.e., $v_k(x)=b_k(x)$ for all $k$, is a socially optimal strategy.
\end{corollary}
By contraposition, we immediately have the following corollary:
\begin{corollary} \label{corollary:notoptimal}
If $v_k(x)$ is not equal to $b_k(x)$ \textcolor{black}{for some $k$}, the allocation based on \eqref{WDP_ps}, \eqref{WDP_pn}, or \eqref{WDP_pw} \textcolor{black}{may} not \textcolor{black}{be} socially optimal.
\end{corollary}

From \eqref{totalcharge}, we can see that the charge received by $k$ does not depend on its declared bid $\bold{b}_k$, but the private information of all other bidders, i.e., $\bold{b}_{-k}$. In other words, it will not charge with a higher amount if it submits bid $b_k(\mathcal{S})>v_k(\mathcal{S})$. Thus there is no incentive for a rational bidder to announce false information and thus the charging mechanism is incentive compatible. 
This can be realized with the following theorem:
\begin{theorem} 
Submitting  bids as the true valuations is a dominant strategy under the VCG-based charging mechanism.
\end{theorem}
\begin{proof}
%
%
%
%
For the splittable service, the auctioneer determines the winners by solving \eqref{WDP_ps}, and thus
\begin{align}
	x(\text{b}) \in \arg\min_{x\in \mathcal{X}} \left(\sum_k b_k(x)\right) = \arg\min_{x\in \mathcal{X}} \left(\sum_{l\in\mathcal{K}\setminus k}b_l(x)  +b_k(x)\right). \label{utility5}
\end{align}
By Theorem \ref{theorem:social}, the mechanism will result in $x$ based on \eqref{utility4} if bidder $k$ declares $\bold{b}_k=\bold{v}_k$, where $\bold{v}_k=[v_k(\mathcal{S})]_{\mathcal{S}\in\mathcal{Q}}$. Since the declaration of $\bold{b}_k$ does not depend on other bidders and increase in bids will not improve its profit, truth telling is a dominant strategy of bidder $k$. The proof is similar for the non-splittable and private services.
\end{proof}
Therefore, if all other bidders submit their true valuations as bids, one will reveal its true values as bids in the auction. The mechanism can achieve the best economic efficiency and all bidders will behave truthfully resulting in maximizing the social welfare.

The auction exhibits choice-set monotonicity. As we allocate the service from the bidders to the customer, when a new bidder is introduced, all previously existing allocations are still feasible and the introduction of the new bidder creates more allocation choices. 
Moreover, the auction has the property of no negative externalities. When bidder $k$ participates in the auction, it can charge $c_k$, as given in \eqref{totalcharge}, which is always non-negative. If it does not win in the auction, its utility will turn zero. Therefore, our VCG-based charging mechanism in individually rational and hence \textit{strategy-proof} \cite{multiagent}.

Consider the splittable service. Recall that $x^*$ and $p^*$ are the optimal solution of Problem \eqref{WDP_ps} and its objective function value, respectively.
The total charge that the customer should pay is
\begin{align}
\sum_{k\in\mathcal{K}}c_k(\bold{b}_{-k}) 
&= \sum_{k\in\mathcal{K}} \left[ p_{-k}^{*} - \left( p^{*} - \sum_{\mathcal{S}\in \mathcal{Q}}b_k(\mathcal{S})x_k^*(\mathcal{S}) \right) \right]\nonumber\\
&=\sum_{k\in\mathcal{K}}\sum_{\mathcal{S}\in \mathcal{Q}}b_k(\mathcal{S})x_k^*(\mathcal{S}) + \sum_{k\in\mathcal{K}}(p_{-k}^* - p^*)\nonumber\\
&= p^* + \sum_{k\in\mathcal{K}}( p_{-k}^* - p^*). \label{asymtotic}
\end{align}
The mechanism generates the total charge for the customer which is  asymptotically closed to that of the optimal auction, i.e., $p^*$ from solving \eqref{WDP_ps}.
When no single bidder has a significant effect, $p_{-k}^*$ will be very close to $p^*$ and the total charge will tend to the social optimum $p^*$.  Similar results also hold for the non-splittable and private services.

\subsection{\textcolor{black}{Computation}} \label{subsec:computation}

By \cite{IP_NP}, integer programming is in $\mathcal{NP}$ (nondeterministic polynomial time), so as Problems \eqref{WDP_ps}, \eqref{WDP_pn}, and \eqref{WDP_pw} and their $k$-exclusive counterparts. The sizes of these problems \textcolor{black}{grow} with $Q$ and $K$. The capacity of a vehicle, $Q$, is usually fixed and small.  As $K$ is the number of operators interested in the transportation request\footnote{For the splittable service, an operator may elect multiple vehicles to fulfill a transportation request and thus it may submit multiple sets of bids for its vehicles.  This can simply be considered as that it is split into multiple virtual operators, each of which represents one of these vehicles. In practice, it is unusually for an operator to be split into many virtual operators. Moreover, the number of such virtual operators that one can raise is upper bounded by the number of available seats $Q$. The capacity of an AV typically follows the standard of a intermediate-size car, i.e., $Q=5$.}, 
 there are just a few  operators constituting the public transportation system in practice. Thus the problems can still be solved effectively. 

Based on \eqref{totalcharge}, to determine $c_k(\bold{b}_{-k})$ for each $k$, we need to solve \eqref{WDP_ps} and its $k$-exclusive counterpart once. Hence\textcolor{black}{, according to \eqref{asymtotic},} to compute $c^s$ ($c^n$ or $c^p$), we need to solve \eqref{WDP_ps} (\eqref{WDP_pn} or \eqref{WDP_pw}) and its $k$-exclusive counterpart $K$ times. 
\textcolor{black}{For the splittable service, Problem \eqref{WDP_ps} and the $K$ $k$-exclusive counterparts are computationally independent and thus they can be solved in parallel. As to be verified in Section \ref{subsec:comptime}, the computation time required to determine the charge for a practical system is reasonably short. Further, if the problem and its counterparts are solved in parallel, the required computation time will be even shorter. Hence, our approach is computationally feasible. The siutations are similar for the non-splittable and private services.}

\section{Performance Evaluation}\label{sec:simulation}

In this section, we verify the analytical results discussed in the previous sections and evaluate the performance of the charging mechanism. We set the capacity of each vehicle to five, i.e., $Q=5$, by following the standard of a intermediate-size car. We consider the scenarios when sets of operators of different size are interested in serving the transportation request $r$, including having $K$ equal to \textcolor{black}{1,} 5, 10, 30, 50, and 100, which are sufficient to evaluate the performance of a typical AV public transportation system spanning from a small to a large scale. \textcolor{black}{Note that the scenario with $K=1$ is used as a control representing a monopolistic environment.} For each of these scenarios, we generate 100 random cases. Each random case is constructed by generating random bidding functions $b_k(\cdot)$ for the bidders, where the cost of occupying one seat is assumed to be a random number in $(0,1]$. The number of available seats in each elected vehicle available at the time of auctioning is randomly set between one and five.
We perform the simulations on a computer with Intel Core i7-3770 CPU at 3.90 GHz and 16 GB of RAM. They are conducted in the MATLAB environment, where the optimization problems are addressed with YALMIP \cite{yalmip} and CPLEX \cite{cplex}.

\subsection{Total Charges} \label{subsec:totalcharges}

\begin{table}[!t]
\renewcommand{\arraystretch}{1.3}
\caption{Number of unservable cases (among 100).}
\label{tab:unservable}
\centering
\begin{tabular}{c|c|c|c | c|c|c}
\hline\hline
\multirow{2}{*}{$K$} & \multirow{2}{*}{Service}	& \multicolumn{5}{c}{No. of seats required}\\
\cline{3-7}
 &	& 1 & 2 & 3 & 4 & 5\\
\hline
\multirow{3}{*}{1}	& Splittable	 & 0	& 19	& 44	& 65	& 87\\
				& Non-splittable& 0	& 19	& 44	& 65	& 87	\\
				& Private	 	 & 87& 87	& 87	& 87	& 87\\
\hline
\multirow{3}{*}{5}	& Splittable	 & 0	& 0	& 0	& 0	& 0\\
				& Non-splittable& 0	& 0	& 6	& 34	& 74	\\
				& Private	 	 & 74& 74	& 74	& 74	& 74\\
\hline
\multirow{3}{*}{10}	& Splittable	 & 0	& 0	& 0	& 0	& 0\\
				& Non-splittable& 0	& 0	& 6	& 3	& 32	\\
				& Private	 	 & 32& 32	& 32	& 32	& 32\\
\hline
\multirow{3}{*}{30}	& Splittable	 & 0	& 0	& 0	& 0	& 0\\
				& Non-splittable& 0	& 0	& 0	& 0	& 0	\\
				& Private	 	 & 0  & 0	& 0	& 0	& 0\\
\hline
\multirow{3}{*}{50}	& Splittable	 & 0	& 0	& 0	& 0	& 0\\
				& Non-splittable& 0	& 0	& 0	& 0	& 0	\\
				& Private	 	 & 0	& 0	& 0	& 0	& 0\\
\hline
\multirow{3}{*}{100}	& Splittable	 & 0	& 0	& 0	& 0	& 0\\
				& Non-splittable& 0	& 0	& 0	& 0	& 0	\\
				& Private	 	 & 0  & 0	& 0	& 0	& 0\\
\hline\hline
\end{tabular}
\end{table}


\begin{figure}[!t]
	\begin{center}
		\subfigure[Splittable service.]{\label{fig:obj2-s}\includegraphics[width=3.5in]{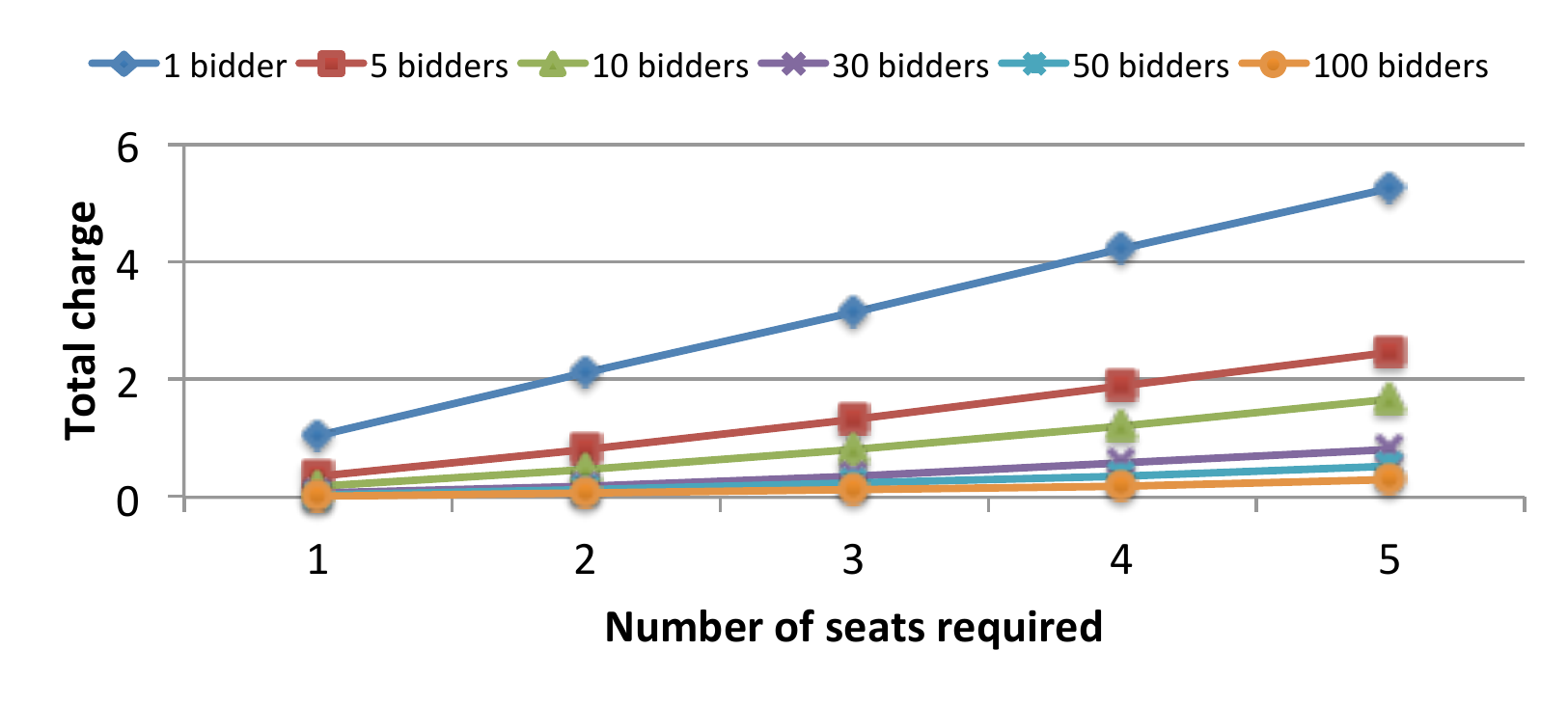}}
    \subfigure[Non-splittable service.]{\label{fig:obj2-n}\includegraphics[width=3.5in]{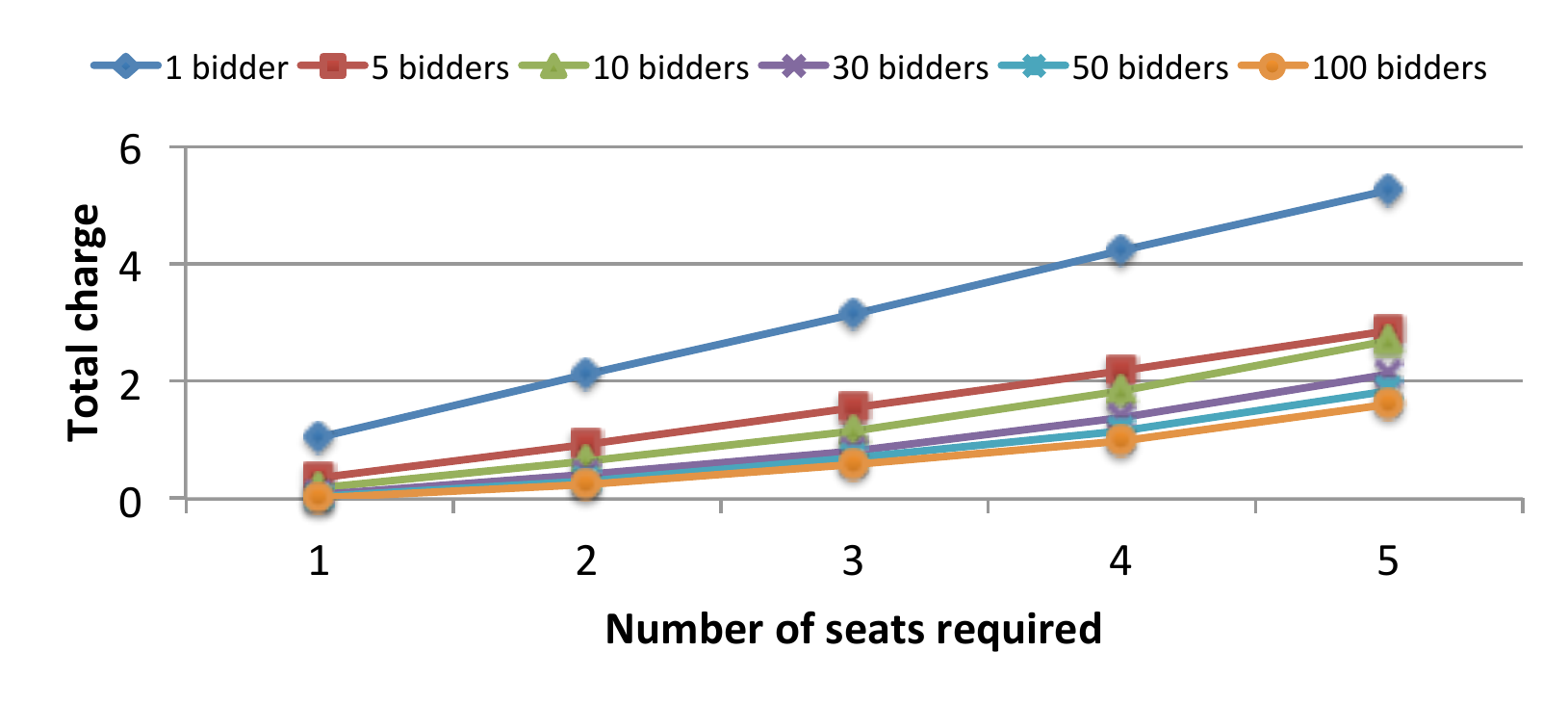}} 
		\subfigure[Private service.]{\label{fig:obj2-w}\includegraphics[width=3.5in]{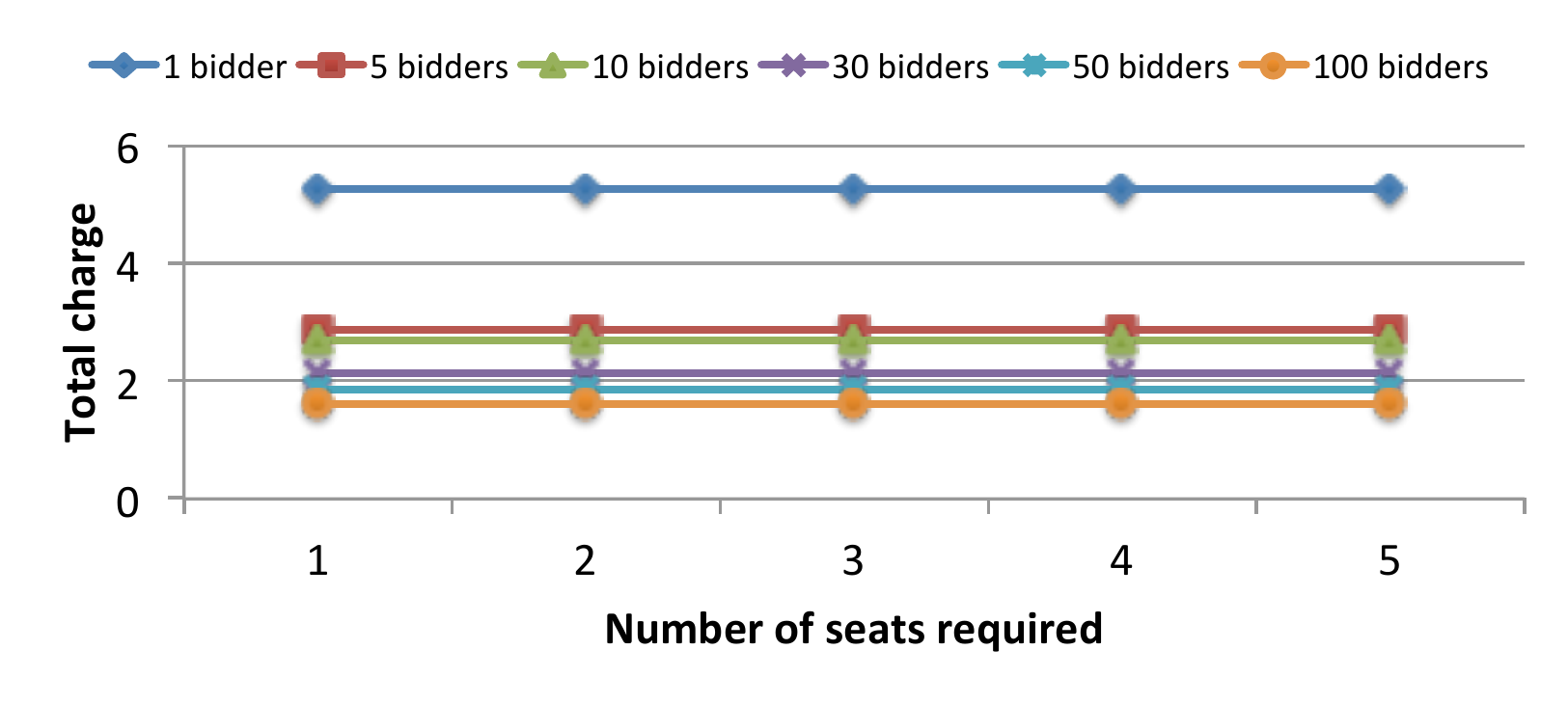}}
	\end{center}
	\caption{Charges computed by the VCG charging mechanism.}
  \label{fig:totalCharge}
\end{figure}

We first examine the servability and the changes of total charges for the three service types. Table \ref{tab:unservable} shows the number of unservable cases among  the 100 cases, for each service type in each scenario. A case is deemed unservable when no operator can elect a vehicle which can meet the occupancy requirement of the request. In general, when the number of bidders increases, the unseravable cases vanish  because the chance of having one vehicle suitable for serving the request becomes higher with a larger set of bidders. When the number of seats required, $q_r$, increases, the number of unservable cases increases due to the more stringent occupancy requirement. For comparing the service types, the private service is the easiest to become unservable while the splittable service is the hardest. As all seats in a vehicle are required, the occupancy requirement for the private service is the most stringent. For the splittable service, it is easier to get suitable vehicles to serve the separated groups of passengers. When the size of bidders is small (e.g., $K=5$ or 10), the customer may only have options from a subset of services among the three service types. With a sufficient number of operators (e.g., $K=30$ or more), the customer can select among all the three services. 
\textcolor{black}{Moreover, the multi-tenant ($K>1$) sencarios have relatively far fewer unservable cases than the monopolistic ($K=1$) one. This shows that multi-tenancy can improve the servability of the system.}

Fig. \ref{fig:totalCharge} shows the charges for the three service types computed from the VCG charging mechanism with different numbers of seats required, $q_r$, stated in the request and each data point is the average of the servable cases. 
\textcolor{black}{As the VCG mechanism is not applicable for $K=1$, the total charges for the single bidder cases are set to the corresponding optimal charges, i.e., $p^*$.}
For the splittable service, as shown in \ref{fig:obj2-s}, the total charge increases with $q_r$. The total cost and the growth rate of the cost in $q_r$ decrease with $K$. A lager $\mathcal{K}$ results in a  stronger competition, which leads to a lower cost. The situation is similar for the non-splittable service given in Fig. \ref{fig:obj2-n}, but the charge differences between the scenarios of different $K$ become smaller when compared with those shown in Fig. \ref{fig:obj2-s}. In other words, the saving from a stronger competition due to  non-splittable service diminishes when compared with the splittable service. The reason is that the flexibility of allocating vehicles to the customer is lower and thus the customer has smaller bargaining power on the bidders. As shown in Fig. \ref{fig:obj2-w}, the total charge for private service is independent of the number of seats required as all seats in the vehicle are required. 
\textcolor{black}{As the monopolistic (1 bidder) senario produces much higher charges, multi-tenancy can increase market competition leading to lower service charges.}

\begin{figure}[!t]
\hspace{-0.3cm}
\includegraphics[width=3.7in]{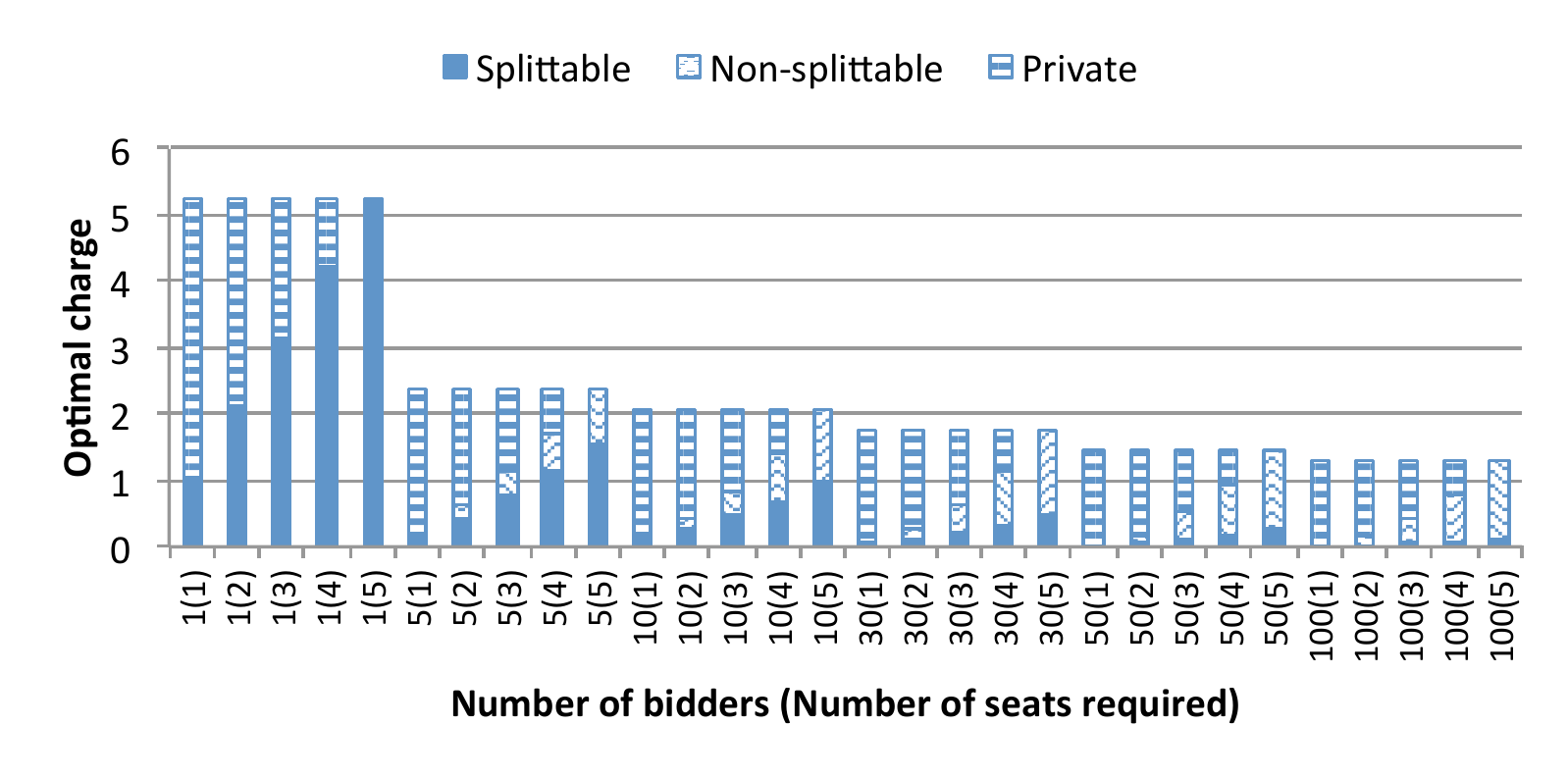} \vspace{-0.5cm}
\caption{Optimal charges for the three service types.} 
\label{fig:opt}
\vspace{-0.5cm}
\end{figure}

The charges induced from the VCG charging mechanism will approach the optimal charges, i.e., $p^s$, $p^n$, and $p^p$,  when no single bidder has significant effect (more results on charge asymptoticity will be discussed in Section \ref{subsec:asym}). Fig. \ref{fig:opt} illustrates the optimal charges and the results are the averages of the servable cases. For clear presentation, the bars of the three service types are overlapped for each scenario. Among the three service types, the splittable service incurs the lowest charge while the private service acquires the highest. This confirms the results given in Theorem \ref{lemma:pricerank}. Similar to Fig. \ref{fig:totalCharge}, a stronger competition (with more bidders) results in a lower charge. When the number of seats required is small, the charges for splittable and non-splittable services are of similar values and they are much lower than that for private service. If the passengers do not mind to share the ride with other people, the non-splittable service is worthwhile to take. When the demand for seat occupancy grows, the difference between the charges for non-splittable and private services decreases.  Hence, the private service becomes more worthwhile.


\subsection{Truthfulness}

\begin{table*}[!t]
\caption{Winners and their charges for the splittable service.}
\label{tab:fool}
\centering
\begin{tabular}{c|c|c|c|c | c|c|c}
\hline\hline
$q_r$ 	& bid 		&\multicolumn{5}{c|}{Winner's ID}&Total \\ 
	& increase	& \multicolumn{5}{c|}{(Charge)}& charge \\\hline

\multirow{6}{*}{1}	& \multirow{2}{*}{0\%} 	& \multicolumn{5}{c|}{69} 				& \multirow{2}{*}{0.0295}\\
										& 											& \multicolumn{5}{c|}{(0.0295)} 	& \\ \cline{2-8}
										& \multirow{2}{*}{20\%}	& \multicolumn{5}{c|}{69}				& \multirow{2}{*}{0.0295}\\
										& 											& \multicolumn{5}{c|}{(0.0295)} 	& \\ \cline{2-8}
										& \multirow{2}{*}{50\%}	& \multicolumn{5}{c|}{82}				& \multirow{2}{*}{0.0363}\\
										& 											& \multicolumn{5}{c|}{(0.0363)} 	& \\ \hline
										
\multirow{6}{*}{2}	& \multirow{2}{*}{0\%} 	& \multicolumn{2}{c}{69} & \multicolumn{3}{c|}{82}				& \multirow{2}{*}{0.0726}\\
										& 											& \multicolumn{2}{c}{(0.0363)} & \multicolumn{3}{c|}{(0.0363)}			 	& \\ \cline{2-8}
										& \multirow{2}{*}{20\%}	& \multicolumn{2}{c}{69} & \multicolumn{3}{c|}{82}				& \multirow{2}{*}{0.0726}\\
										& 											& \multicolumn{2}{c}{(0.0363)} & \multicolumn{3}{c|}{(0.0363)} 	& \\ \cline{2-8}
										& \multirow{2}{*}{50\%}	& \multicolumn{2}{c}{60} & \multicolumn{3}{c|}{79}				& \multirow{2}{*}{0.1049}\\
										& 											& \multicolumn{2}{c}{(0.0524)} & \multicolumn{3}{c|}{(0.0524)} 	& \\ \hline
										
\multirow{6}{*}{3}	& \multirow{2}{*}{0\%} 	& \multicolumn{1}{c}{69}	&\multicolumn{2}{c}{79} & \multicolumn{2}{c|}{82} 				& \multirow{2}{*}{0.1444}\\
										& 											& \multicolumn{1}{c}{(0.0481)}	&\multicolumn{2}{c}{(0.0481)} & \multicolumn{2}{c|}{(0.0481)} 	& \\ \cline{2-8}
										& \multirow{2}{*}{20\%}	& \multicolumn{1}{c}{69}	&\multicolumn{2}{c}{79} & \multicolumn{2}{c|}{82}				& \multirow{2}{*}{0.1444}\\
										& 											& \multicolumn{1}{c}{(0.0481)}	&\multicolumn{2}{c}{(0.0481)} & \multicolumn{2}{c|}{(0.0481)} 	& \\ \cline{2-8}
										& \multirow{2}{*}{50\%}	& \multicolumn{1}{c}{38}	&\multicolumn{2}{c}{41} & \multicolumn{2}{c|}{60}				& \multirow{2}{*}{0.2133}\\
										& 											& \multicolumn{1}{c}{(0.0711)}	&\multicolumn{2}{c}{(0.0711)} & \multicolumn{2}{c|}{(0.0711)} 	& \\ \hline
									
\multirow{6}{*}{4}	& \multirow{2}{*}{0\%} 	& \multicolumn{1}{c}{60} 		& \multicolumn{1}{c}{69} & \multicolumn{1}{c}{79} &\multicolumn{2}{c|}{82}		& \multirow{2}{*}{0.2097}\\
										& 											& \multicolumn{1}{c}{(0.0524)} & \multicolumn{1}{c}{(0.0524)} & \multicolumn{1}{c}{(0.0524)} & \multicolumn{2}{c|}{(0.0524)}	& \\ \cline{2-8}
										& \multirow{2}{*}{20\%}	& \multicolumn{1}{c}{41} & \multicolumn{1}{c}{69} & \multicolumn{1}{c}{79} & \multicolumn{2}{c|}{82}				& \multirow{2}{*}{0.2310}\\
										& 											& \multicolumn{1}{c}{(0.0577)} & \multicolumn{1}{c}{(0.0577)} & \multicolumn{1}{c}{(0.0577)} & \multicolumn{2}{c|}{(0.0577)} 	& \\ \cline{2-8}
										& \multirow{2}{*}{50\%}	& \multicolumn{1}{c}{29}	& \multicolumn{2}{c}{38}& \multicolumn{2}{c|}{41} & \multirow{2}{*}{0.3326}\\
										& 											& \multicolumn{1}{c}{(0.1904)} &  \multicolumn{2}{c}{(0.0711)} & \multicolumn{2}{c|}{(0.0711)} 	& \\ \hline
										
\multirow{6}{*}{5}	& \multirow{2}{*}{0\%} 	& \multicolumn{1}{c}{41} & \multicolumn{1}{c}{60} & \multicolumn{1}{c}{69} & \multicolumn{1}{c}{79} & \multicolumn{1}{c|}{82}				& \multirow{2}{*}{0.3115}\\
										& 											& \multicolumn{1}{c}{(0.0623)} & \multicolumn{1}{c}{(0.0623)} & \multicolumn{1}{c}{(0.0623)} & \multicolumn{1}{c}{(0.0623)} & \multicolumn{1}{c|}{(0.0623)} 	& \\ \cline{2-8}
										& \multirow{2}{*}{20\%}	& \multicolumn{1}{c}{38}	& \multicolumn{1}{c}{60} & \multicolumn{1}{c}{69} & \multicolumn{1}{c}{79} & \multicolumn{1}{c|}{82}			& \multirow{2}{*}{0.3146}\\
										& 											& \multicolumn{1}{c}{(0.0629)} & \multicolumn{1}{c}{(0.0629)} & \multicolumn{1}{c}{(0.0629)} & \multicolumn{1}{c}{(0.0629)} & \multicolumn{1}{c|}{(0.0629)} 	& \\ \cline{2-8}
										& \multirow{2}{*}{50\%}	& \multicolumn{1}{c}{29} & \multicolumn{1}{c}{36} & \multicolumn{1}{c}{38} & \multicolumn{2}{c|}{84}				& \multirow{2}{*}{0.4952}\\
										& 											& \multicolumn{1}{c}{(0.2039)} & \multicolumn{1}{c}{(0.1220)} & \multicolumn{1}{c}{(0.0846)} & \multicolumn{2}{c|}{(0.0846)}	& \\ 
\hline\hline
\end{tabular}
\end{table*}

Next we investigate the consequence of not telling the truth. That is, the bidders intend to raise their bids ($b_k(\mathcal{S})>v_k(\mathcal{S})$) for potential higher profit. 
We illustrate the results based on a particular random example of the $K=100$ scenario used in Section \ref{subsec:totalcharges}.
Table \ref{tab:fool} gives the winners and their charges with increasing bids for the cases of different $q_r$, i.e.,  the number of seats required for the request of splittable service. We consider the cases with zero bid increase as the base cases. For those with positive bid rise, we increase the bids of the winners of their corresponding base cases with certain percentages. From the results with $q_r$ equal to 1, 2, and 3, it is clear that an increase of one's bid does not increase its final charge but possibly the others' charges. When the increase is large (e.g., by $50\%$), it may lead to changes of winners. For example, for $q_r=1$, Bidder 69 is the winner in the base case with charge equal to $0.0295$. An increase of its bid by 20\% does not help charge more (also $0.0295$) while a larger increase by $50\%$ results in the change of winner to Bidder 82 with a large charge of $0.0363$. For $q_r$ equal to 4 and 5, it seems that some bidders do improve their charges by raising their bids. For example, for $q_r=4$, Bidder 69 can increase its final charge from 0.0524 to 0.0577 by raising its bid for 20\%. However, this charge improvement does not come from Bidder 69's bid increase but it is due to Bidder 60's action. As seen from \eqref{totalcharge}, the charge $c_k(\bold{b}_{-k})$ that a bidder $k$ receives from the auction is independent of its submitted bid $\bold{b}_k$. Hence, the charge improvement due to one's bid rise is just an illusion.
As a whole, an increase of one's bid does not improve, or even depresses, its utility and this potentially enhances the utilities of other bidders only. Moreover, the customer may need to pay more in the presence of untruthful bidders.

\begin{figure*}[!t]
	\begin{center}
		\subfigure[Splittable service.]{\label{fig:fool-s}\includegraphics[width=7.2in]{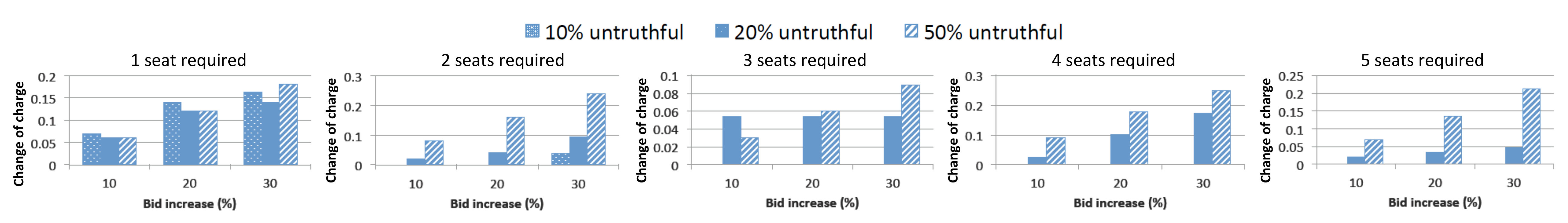}}
    \subfigure[Non-splittable service.]{\label{fig:fool-n}\includegraphics[width=7.2in]{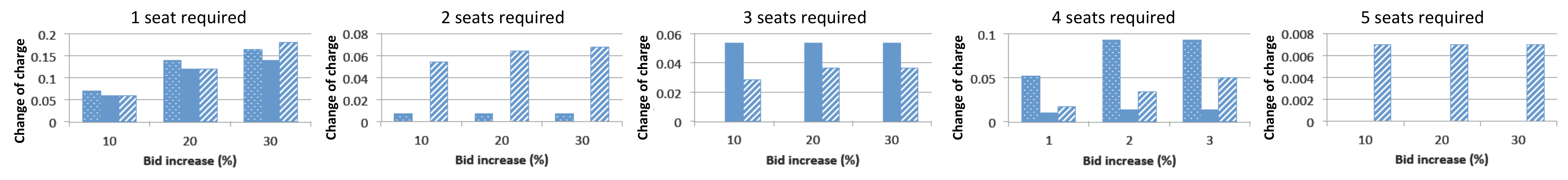}} 
		\subfigure[Private service.]{\label{fig:fool-w}\includegraphics[width=7.2in]{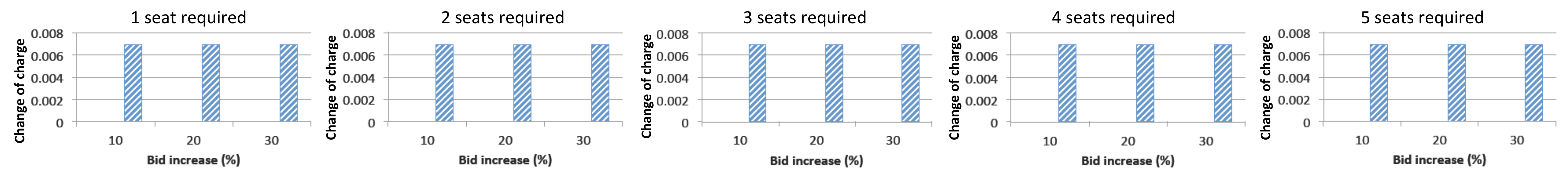}}
	\end{center}
	\caption{Change of charge with Truthfulness.}
  \label{fig:fool}
\end{figure*}
We further examine the total charges, instead of individual charges, for the three service types.
We define the ``change of charge'' as $\frac{\alpha - \beta}{\beta}$, where $\alpha$ and $\beta$ are the total charges with and without untruthfulness, respectively.
The latter refers to the total charge when all bidders report their bids as the true valuations, i.e., $b_k(\mathcal{S})=v_k(\mathcal{S})$ for all $\mathcal{S}$ and $k$. With untruthfulness, some bidders raise their bids higher than the valuations. Positive change of charge means that the customer needs to pay more as a whole (not for a particular operator).
Fig. \ref{fig:fool} gives the change of charge with different number of untruthful bidders and different amount of bid rise. For each simulation run, we randomly select 10\%, 20\%, or 50\% of the bidders as untruthful bidders, who raise their bids from the true valuations by 10\%, 20\%, or 30\%. In each case, the change of charge is always nonnegative. For those without change of charge, the bidders selected for bid increase are not the winners of the auction. Their rise of bids will not affect the charges to the customer. In reality, a bidder will not know if it would be a winner before conducting the auction. So a rise of one's bid may not always make the system performance worse. For the splittable service, the total charge tends to increase with higher bid rise. However, the trends for the non-splittable and private services are not obvious. As there is only one winner in the auction for these two services, the change of charge becomes sensitive to the selected group of untruthful bidders.

\subsection{Asymptoticity of the Charges} \label{subsec:asym}

\begin{figure}[!t]
\centering
	\begin{center}
		\subfigure[Splittable service.]{\label{fig:var-s}\includegraphics[width=3.5in]{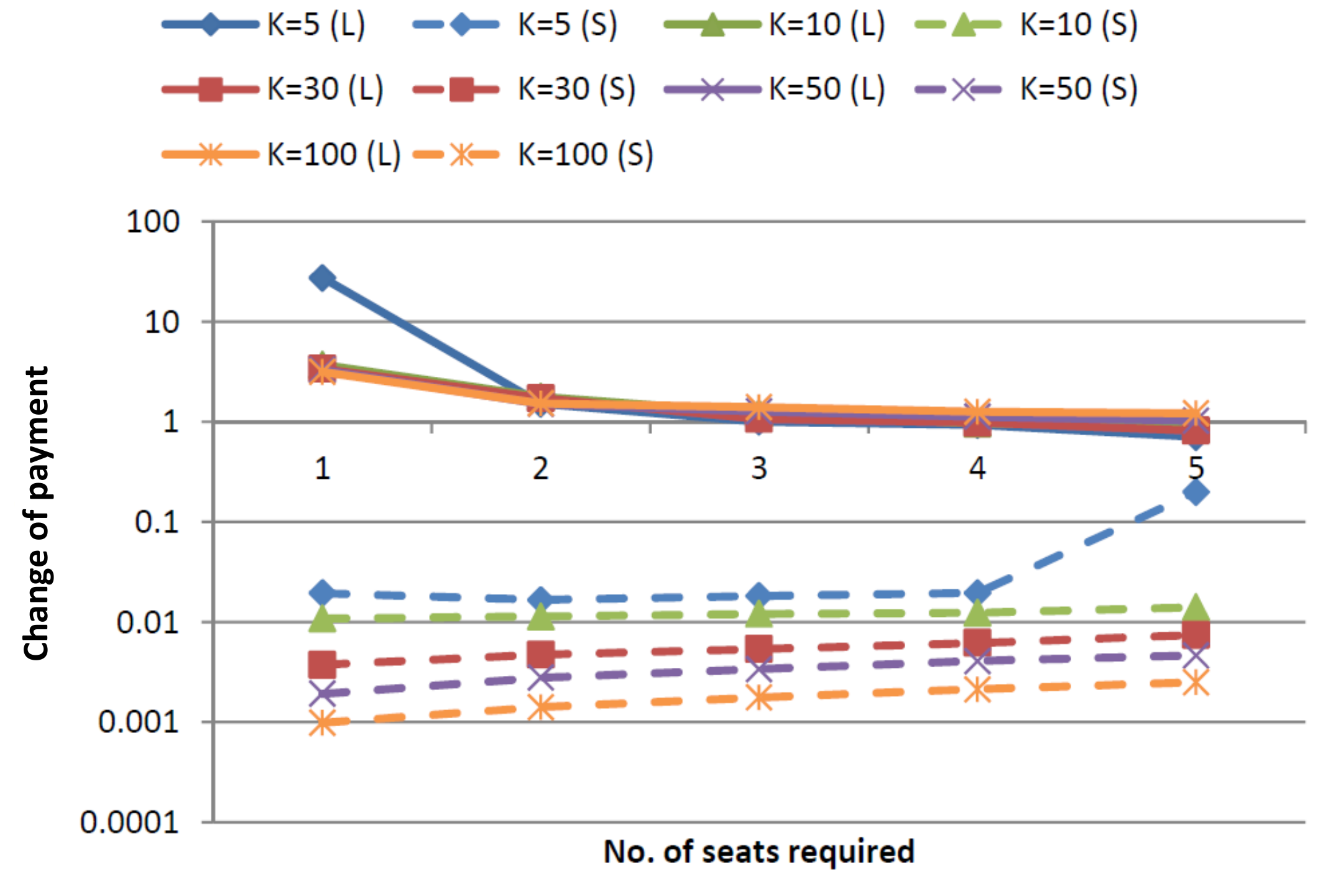}}\\
    \subfigure[Non-splittable service.]{\label{fig:var-n}\includegraphics[width=3.5in]{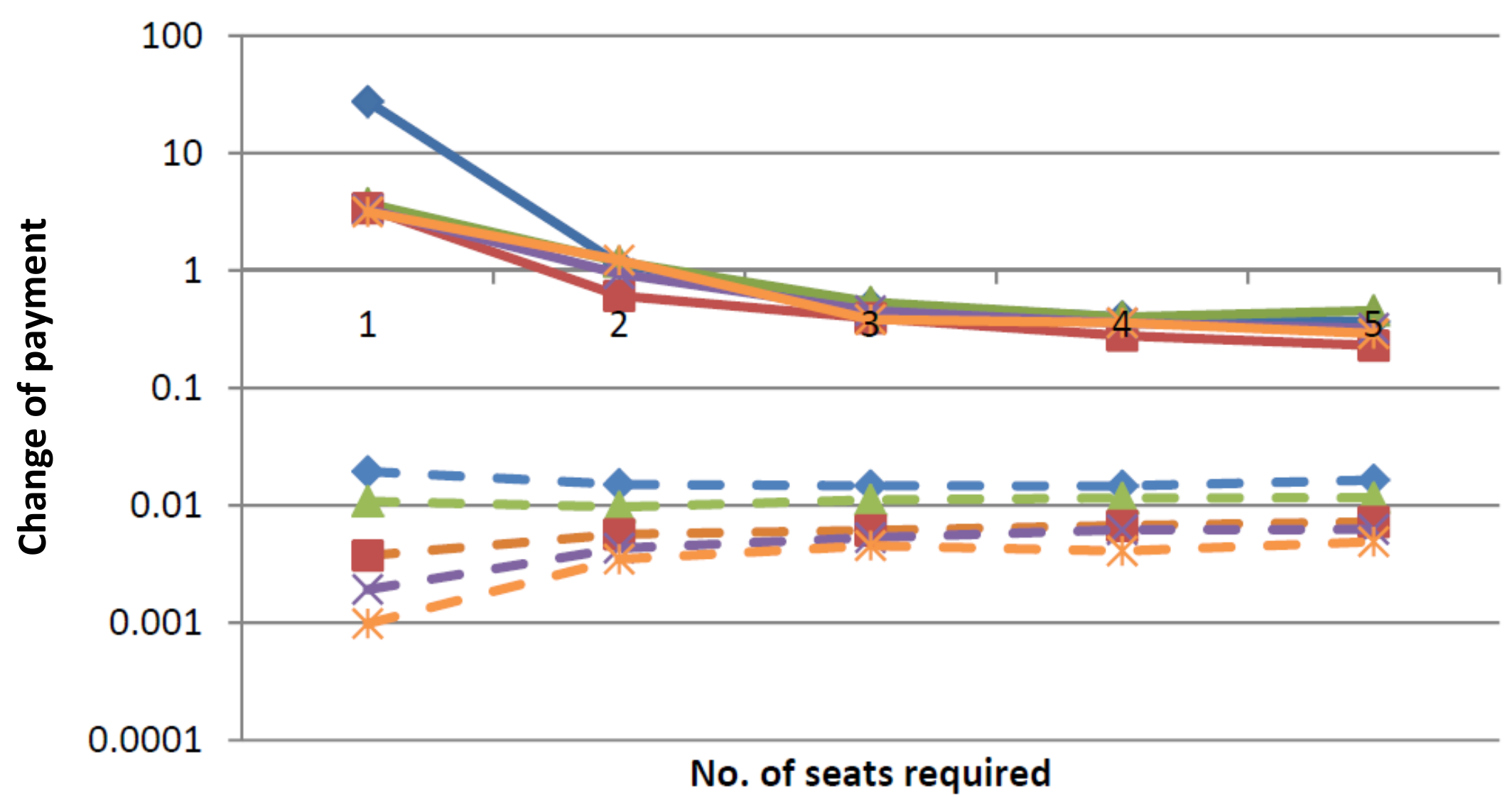}} 
		\subfigure[Private service.]{\label{fig:var-w}\includegraphics[width=3.5in]{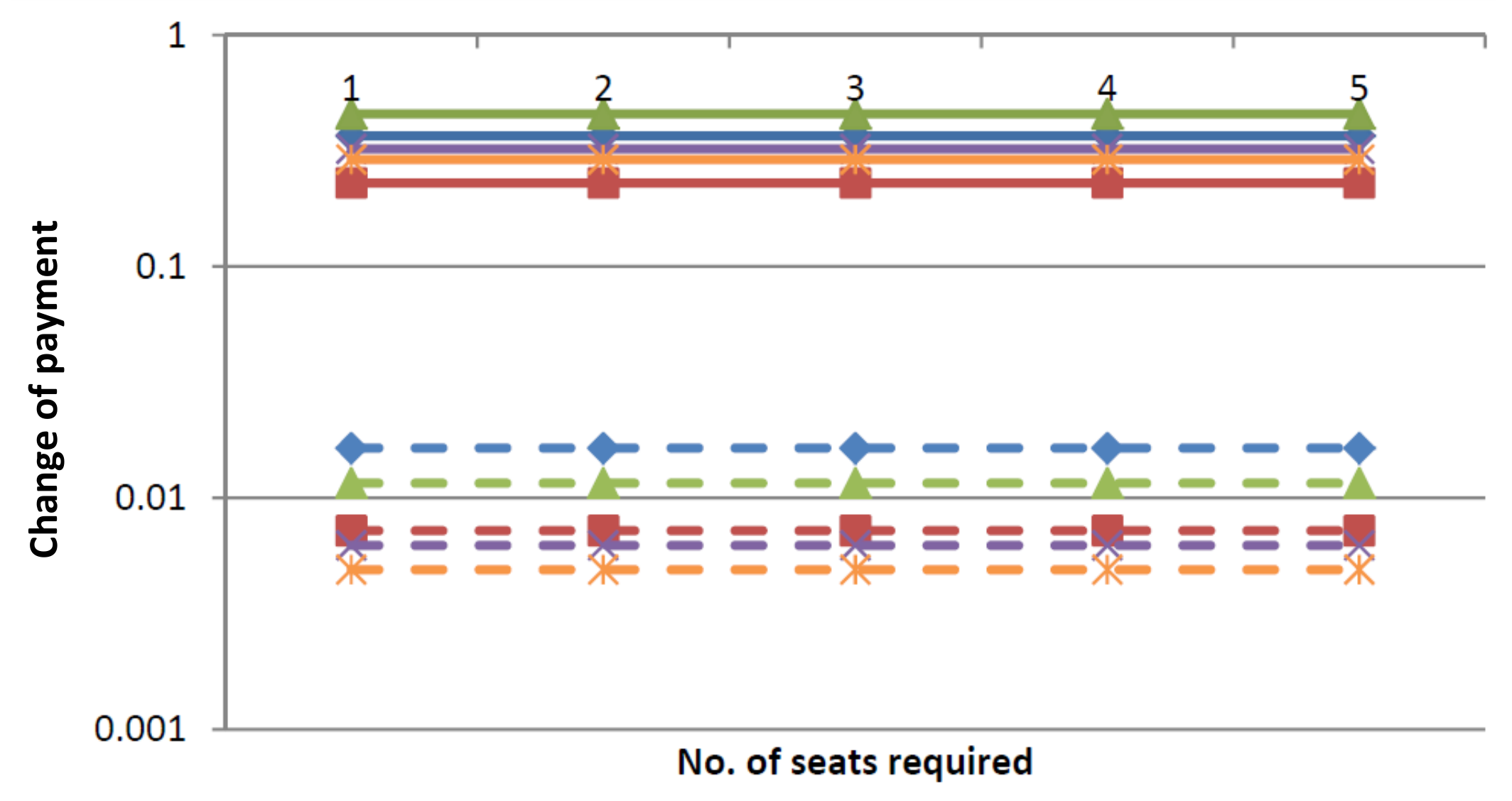}}
	\end{center}
	\caption{Change of payment with different variations of bids.}
  \label{fig:var}
\end{figure}

We interrogate the asymptoticity of the charges as illustrated in \eqref{asymtotic}. We define ``change of payment'' as $\frac{\sum_k c_k - p^*}{p^*}$ (based on \eqref{asymtotic}). The smaller the change of payment, the smaller the difference between the total charge and the total bid of the optimal auction is. We consider two ways of generating the random bids, i.e., with large and small variations. For the former, all bids for occupying one seat are randomly generated in $(0,1]$. For the latter, the bids are constructed as $rand(0,0.1]+0.5$, where $rand(0,0.1]$ refers to a random number in $(0,0.1]$. So the values of bids for each seat occupancy with small variations are closer to each other. The results for different $K$ and $q_r$ are provided in Fig. \ref{fig:var} where the change of payment is presented in the logarithmic scale and the large and small variations are indicated as ``(L)'' and ``(S)'', respectively.\footnote{\textcolor{black}{As the VCG mechanism is not applicable for the monopolistic scenario, here we focus on those with $K>1$.}} We compare the results of the two variations in pair for each case and bids with smaller variations always result in smaller change of payment. This confirms our interpretation of \eqref{asymtotic} in Section \ref{sec:charging}, i.e., the total charge is close to the result of  the optimum auction when no single bidder has a significant effect.  It is unlikely to have a bidder affect the auction significantly when the values of bids are close to each other.

\subsection{Computation Time} \label{subsec:comptime}
\begin{figure}[!t]
\centering
\hspace{-0.3cm}
\includegraphics[width=3.5in]{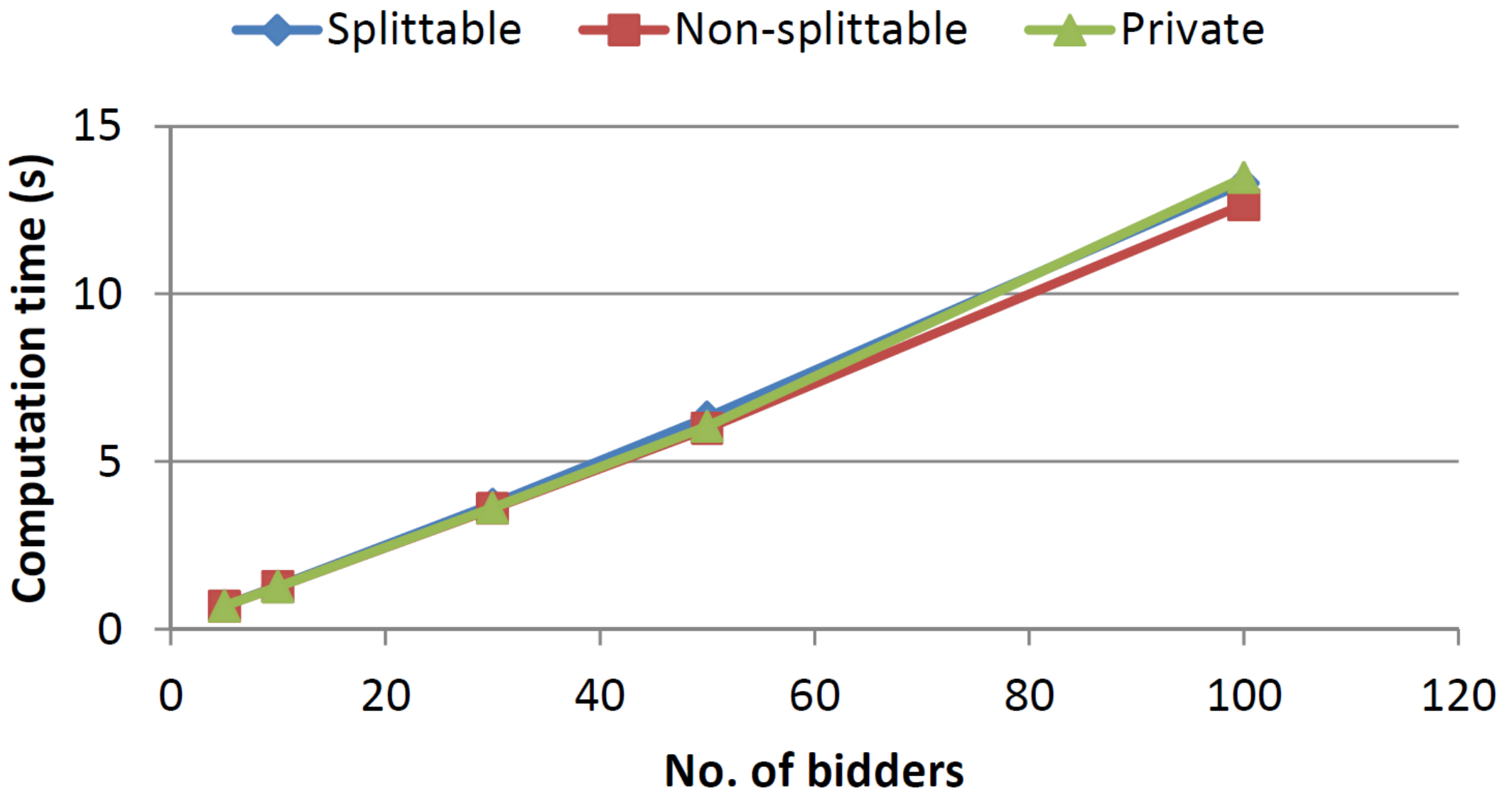} \vspace{-0.5cm}
\caption{Computation time.} 
\label{fig:time}
\vspace{-0.5cm}
\end{figure}

We investigate the computation time for computing the charges with the scenarios of $K$ equal to 5, 10, 30, 50, and 100. Fig. \ref{fig:time} shows the computation times required to determine the total charges for the three service types, where each data point is the average of 100 random cases. In general, the three services require similar computation time, which is proportional to $K$. 
\textcolor{black}{Alhough our approach is designed to handle one service request, there may be multiple outstanding requests needed to be processed at the same time in a pratical system. In fact, our approach can facilitate two levels of parallelization to make the computation efficient. First, as illustrated in Section \ref{subsec:pricingprocess}, we can simply apply the pricing process to each of the requests and obtain the corresponding charges simultaneously. Second, as discussed in Section \ref{subsec:computation}, the pricing process can be further broken into multiple computational modules, each of which can be run in parallel, due to the fact that Problem \eqref{WDP_ps} (\eqref{WDP_pn} or \eqref{WDP_pw}) and its $k$-exclusive counterparts are computationally independent.}
When higher performance is required in a practical environment, instead of a general-purpose computer, a high-performance field-programmable gate array-based computer \cite{FPGA}  can be utilized to boost the performance on computing the charges.
\textcolor{black}{Hence, the problems can be solved in parallel and the computation of a total charge requires much less time. This suggests that our approach can be made efficient to implement in a practical system.}

\section{Conclusion}\label{sec:conclusion}

To improve people's standard of living, we may turn a city into a smart city through smarter utilization of resources with modern technologies. Transport is one of the main foci in smart city research and development. With the advent of AV technologies, AVs will become prevalent in the future transportation system. Recently, a novel AV-based public transportation system has been proposed and it is capable of providing new forms of transportation services with high efficiency, high flexibility, and low cost. For the benefit of passengers, multitenancy can increase market competition leading to lower service charge and higher quality of service. This paper is dedicated to studying the pricing issue of the multi-tenant AV public transportation system supported by multiple AV operators. We broaden the service options by introducing three types of services, including the splittable, non-splittable, and private services. We model the pricing process with combinatorial auctions, where the operators bid for offering the service requested by a customer. For each service type, the winners of the auction are determined through an integer linear program. To prevent the bidders from lying about the true valuations in their bids, we propose a VCG-based charging mechanism, which is strategy-proof and capable of maximizing the social welfare, to settle the final charges. We verify our analytical results and evaluate the performance of the charging mechanism with extensive simulations. We summarize the contributions achieved in this paper as follows: (i) introducing multitenancy to the AV public transportation system; (ii) identifying the splittable, non-splittable, and private service types; (iii) modeling the pricing process with combinatorial auctions; (iv) formulating the winner determination problems as integer linear programs; (v) proposing a strategy-proof charging mechanism; (vi) deriving analytical results for the pricing process and the charging mechanism; and (vii) verifying the results with extensive simulations. In the future, we will study a coalition game-based pricing for the system, in which the operators may exchange private information for higher profit.




%
%
%


\bibliographystyle{IEEEtran}
\bibliography{IEEEabrv}


\end{document}